\pdfoutput=1
                
\documentclass[a4paper, USenglish, runningheads]{scrartcl}

\usepackage[utf8]{inputenc} 

\usepackage{amsmath,amssymb,amsthm}
\usepackage{mathtools}

\usepackage{color}
 
\usepackage[activate,final]{microtype} 
\usepackage{xspace}
\usepackage{makecell}

\usepackage{etoolbox}

\usepackage{enumitem}
\setlist[itemize]{leftmargin=*, topsep=5pt}
\setlist{nolistsep}
\setlist[enumerate]{leftmargin=*, topsep=5pt}
\usepackage[english=british,german=guillemets]{csquotes}
\usepackage{booktabs}\usepackage{flafter}
\usepackage{bm}

\usepackage{xparse}
\usepackage{cases}
\newtheorem{theorem}{Theorem}
\newtheorem{corollary}[theorem]{Corollary}
\newtheorem{definition}[theorem]{Definition}
{\bfseries}{\itshape}
\newtheorem{llemma}[theorem]{Lemma}{\bfseries}{\itshape}

\usepackage{tikz}
\usetikzlibrary{calc,arrows,decorations,decorations.pathreplacing,backgrounds,shapes,fit}

\usepackage[activate,final]{microtype}

\usepackage[pdfencoding=auto, psdextra]{hyperref}
\usepackage{xcolor}
\hypersetup{
    colorlinks,
    linkcolor={red!50!black},
    citecolor={blue!80!black},
    urlcolor={blue!80!black}
}

\usepackage{bm}
\usepackage{makecell}

\usepackage{mathtools}
\usepackage{cases}

\DeclarePairedDelimiter{\size}{\lvert}{\rvert}
\DeclarePairedDelimiter{\p}{\lparen}{\rparen}

\DeclarePairedDelimiter{\cb}{\{}{\}}
\DeclarePairedDelimiter{\br}{[}{]}

\newcommand{\eq}[2][]{\begin{equation}\label{#1}\begin{split}#2\end{split}\end{equation}}
\newcommand{\eqw}[2][]{\begin{equation*}\label{#1}\begin{split}#2\end{split}\end{equation*}}

\let\phi\varphi

\let\emptyset\varnothing

\newcommand{\pclique}{\textsc{\mdseries $p$-Clique}}
\newcommand{\pdoms}{\textsc{\mdseries $p$-Dominating-Set}}
\newcommand{\csp}{\textsc{\mdseries CSP}}
\newcommand{\pwcspf}[1][]{\textsc{\mdseries $p$-W\csp$\p{#1}$}}
\newcommand{\pwcsp}{\textsc{\mdseries $p$-W\csp}}
\newcommand{\pscsp}{\textsc{\mdseries $p$-Size-\csp}}

\newcommand{\phhs}{\textsc{\mdseries $p$-W-Hypergraphs-Hitting-Set}}
\newcommand{\pss}{\textsc{\mdseries $p$-Subset-Sum}}
\newcommand{\subsetsum}{\textsc{\mdseries Subset-Sum}}
\newcommand{\pssf}[1][]{\textsc{\mdseries $p$-Subset-Sum$\p{#1}$}}
\newcommand{\cnf}{\textsc{\mdseries CNF}}

\newcommand{\perfectcode}{\textsc{\mdseries $p$-Perfect-Code}}
\newcommand{\pewsat}[1][]{\textsc{\mdseries $p$-Exact-WSat\ensuremath{\p{#1}}}}
\newcommand{\satof}[1][]{\textsc{\mdseries Sat\ensuremath{\p{#1}}}}

\newcommand{\N}{\mathbb{N}}

\usepackage{colonequals}
\newcommand{\defeq}{\colonequals}

\NewDocumentCommand{\W}{o}{\textsc{\mdseries W}\IfValueT{#1}{\ensuremath{\br{#1}}}}

\title{Uniform \csp\ Parameterized by \\ Solution Size is in \W[1] }
\author{Ruhollah Majdoddin\thanks{Institut für Informatik, Humboldt Universität zu Berlin}
\href{https://orcid.org/0000-0003-0711-6287}{\protect\includegraphics[width=4mm, height=3mm]{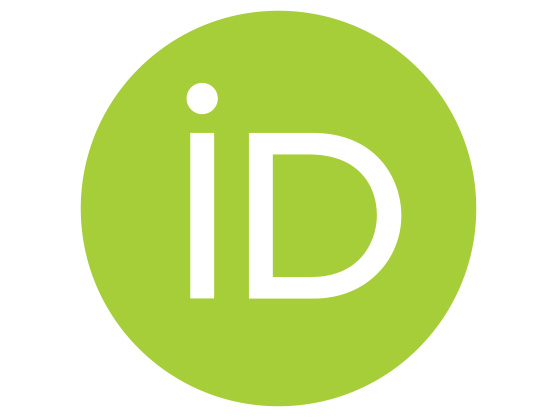}}
\\r.majdodin@gmail.com}

\date{}
\begin{document}

\title{Uniform \csp\ Parameterized by \\ Solution Size is in \W[1] }

%

%


%
\maketitle

\begin{abstract}
		We show that the uniform Constraint Satisfaction Problem (\csp ) parameterized by the size of the solution is in $\W[1]$ 
	(the problem is \W[1]-hard and it is easy to place it in $\W[3]$).
        Given a single ``free'' element of the domain, denoted by $\bm 0$,   
we define the \textit{size} of an assignment as the number of variables that are mapped to a value other than $\bm 0$.
        Named by Kolaitis and Vardi (2000), \textit{uniform} \csp\
        means that the input contains the domain and the list of tuples of each relation in the instance.
	Uniform \csp\ is polynomial time equivalent to homomorphism problem and also to evaluation of conjunctive queries on relational databases.
        It also has applications in artificial intelligence. 	

	We do not restrict the problem to any (finite or infinite) family of relations.
Marx and  Bulatov (2014) showed that Uniform \csp\ restricted to some finite family of relations (thus with a bound on  the arity of relations) 
 and over any finite domain is either \W[1]-complete or fixed parameter tractable.
 
 We then prove that parameterized \subsetsum\ with weights bounded by $n^k$ is in \W[1]. Abboud et al. (2014) have already proved it, but our proof is
much shorter and arguably more intuitive.

Lastly, we study the weighted \csp\ over the Boolean domain, where each variable is assigned a weight, and given a target value, 
it should be decided if there is a satisfying assignment of size $k$ (the parameter) such that the weight of its 1-variables adds up to the target value. 
We prove that if the weights are bounded by $n^{k}$, then the problem is in \W[1]. 

Our proofs give a nondeterministic RAM program with special properties deciding the problem.
        First defined by Chen et al. (2005), such programs characterize \W[1].

\end{abstract}

\section{Introduction}
The Constraint Satisfaction Problem (\csp ) is a fundamentally important problem in computer science, that can express a large number of problems in artificial intelligence and operational research \cite{RBW06}.
An instance $I$ of  \csp\ is  specified by a finite domain $D$, a set of relations over domain $D$, a set $V$ of variables, and a set of constraints $C$ of
 the form $R\p{x_1, \ldots, x_r}$, where $R$ is one of the relations with arity $r \ge 1$ and $x_1, \ldots, x_r \in V$.
 An \textit{assignment} to a set of variables $S\subseteq V$ is a mapping from $S$ to $D$. 
 An assignment to the set of variables of a constraint \textit{satisfies} the constraint
 if evaluating the tuple of variables of the constraint according to the assignment, gives a tuple in the corresponding relation.
 An assignment to $V$ is a \textit{satisfying assignment} of $I$, if it satisfies all the constraints in $I$.
 When seen as a decision problem, the question in \csp\ is whether the given instance has a satisfying assignment.
 
Kolaitis and Vardi \cite{KV00} made a distinction between \textit{nonuniform} \csp ,  
where the domain and the family of relations are fixed, and \textit{uniform} \csp , 
where the input contains the domain and the list of tuples of each relation in the instance. 
They showed that uniform \csp\ is polynomial time equivalent to evaluation of conjunctive queries on relational databases.
Feder and Vardi \cite{FV98} observed that uniform \csp\ (which had already applications in artificial intelligence) and the \textit{homomorphism problem} are polynomial time equivalent.
 
We study uniform \csp\ in the settings of parameterized  complexity. 
 Given a single ``free'' element of the domain, denoted by $\bm 0$,   
 we define the \textit{size} of an assignment as the number of variables that are mapped to a value other than $\bm 0$.
The \textit{Parameterized Size \csp} is defined as follows:
\begin{center}
\begin{tabular}{|l|}
\hline
\begin{tabular}{l}
        \pscsp \rule{0pt}{15pt} \\
\end{tabular}\\
\begin{tabular}{r l}
\textit{Instance:} &\makecell[tl]{
A domain $D$ including $\bm 0$, a set of variables $V$, a set of \ \  \\ constraints, list of tuples of each relation, and $k \ge 0$.
}  \\
\textit{Parameter:} &$k$.\\
\textit{Problem:} \rule[-15pt]{0pt}{15pt}&\makecell[tl]{Decide whether there is a satisfying assignment of size      $k$.
\rule[-10pt]{0pt}{10pt}}\\
 \end{tabular}\\
\hline
\end{tabular}
\end{center}
Many parameterized problems ask, given a structure $\cal A$ (on universe $A$), 
if there is a set $S \subset A$  of a given cardinality (the parameter) such that 
the substructure induced by $S$ has a special property.
Many of these problems can be readily reduced to \pscsp , such that
the size parameter in the resulting \pscsp\ instance has the same value as the cardinality of the set looked for.
A good example is \pclique , which, given an instance $\p{G, k}$, asks if Graph $G$ has a clique of size $k$. 
Another example is $p$-Vertex-Cover. These problems,  however, can be expressed with \pscsp\ restricted to a finite family of relations.
To capture the full expressiveness of \pscsp\ on Boolean domain, we introduce the following problem which is fixed parameter equivalent to \pscsp:
\begin{center}
\begin{tabular}{ |  l |}
        \hline
\begin{tabular}{l}
        \phhs \rule{0pt}{15pt} \\
\end{tabular}\\
\begin{tabular}{r l}
\textit{Instance:} &\makecell[tl]{
	A set $W$, hypergraphs $\p{V_1, E_1}, \ldots, \p{V_m, E_m}$
where $V_i \subseteq W$,\\ and $k \ge 0$.
}  \\
\textit{Parameter:} &$k$.\\
\textit{Problem:} \rule[-15pt]{0pt}{15pt}&\makecell[tl]{Decide whether there is a set $S \subset W$ of cardinality $k$ such that\\
	$S\cap V_i \in E_i$ for  $1 \le i \le m$.
\rule[-10pt]{0pt}{10pt}}\\
\end{tabular}\\
\hline 
\end{tabular}
\end{center}
 
Our main contribution is the following containment theorem:
\begin{theorem}\label{theorem-main}
	\pscsp\  $ \in \W[1]$.
\end{theorem}
\begin{corollary}
$\phhs  \in \W[1]$.
\end{corollary}
We prove the theorem by giving a \textit{tail-nondeterministic $\kappa$-restricted NRAM program} (explained in the next section) deciding the problem.
The significance of our containment result is that it is for the general problem, without restricting it to any (finite or infinite) family of relations.

Our work builds upon the work of Cesati \cite{Ces02}, which, 
answering a longstanding open problem, proved that \perfectcode\ is in $\W[1]$. 
Downey and Fellows \cite{DF95b} had already shown that this problem is $\W[1]$-hard 
and had conjectured that it either represents a natural degree intermediate between $\W[1]$ and $W[2]$, or is complete for $W[2]$.
There is a natural reduction from \perfectcode\ to  \pewsat[\cnf ^+], 
and the proof of \cite{Ces02} can be readily adapted to decide the latter problem. 
This problem is to decide, given a \cnf\ without negation symbols and a natural number $k$,
whether there is an assignment of size $k$, such that exactly one variable in each clause is mapped to $1$.
This can be seen as \pscsp\ restricted to a specific (infinite) family of Boolean relations,
where a tuple is in a relation, if and only if the tuple has exactly one $1$ (this implies that  \pscsp\ is \W[1]-hard).
Notice that because we do not restrict the problem to any family of relations,
our result generalizes that of Cesati  in at least three ways: 
Size of the tuples in the relations are not bounded, 
the (Boolean) relations do not need to be \textit{symmetric} (symmetric means that a tuple being in the relation depends only on the number of $1$s in the tuple),
and an instance can have any finite domain.

In fact, \pewsat[\cnf ^+] is an example of an interesting special case of our containment result:  
\pscsp\ restricted to any (infinite) family of symmetric Boolean relations, 
provided that there is a bound on the size of the tuples of any relation in the family.           
Notice that the bound implies that the number of tuples of each relation is bounded by a polynomial in the arity of the relation. 
Thus, listing all the tuples in the input makes the size of input at most polynomially bigger, and uniform and nonuniform \csp\ in this case have the same complexity.

It is not hard to see that \pscsp\ over the Boolean domain is in \W[3], by reducing it to the \textit{parameterized weighted satisfiability problem} for a class of circuits with bounded depth and \textit{weft} $3$:
one for the conjunction of all constraints in the instance, one for the disjunction of all satisfying assignments of each constraint,
and one to specify each satisfying assignment of a constraint. So what is the significance of placing a problem from (at least) \W[3] down to \W[1]?
First,  although it is a fundamental conjecture that \W[1]-complete problems are not fixed-parameter tractable,
many of them can still be solved substantially faster than exhaustive search over all $\binom{n}{k}$ subsets.
For example, \cite{NP85} gives  an $O\p{n^{.793k}}$ time algorithm for  \pclique .
In contrast, the $\W[2]$-complete problem \pdoms,
was shown by \cite{PW10} not to have such algorithms, unless \satof[\cnf] has an $O\p{2^{\delta n}}$ time algorithm
for some $\delta < 1$, which is an important open problem.
Second, we can express the problems in  \W[1] by a logic that is (conjectured to be) a proper subclass of any logic that can express the problems in \W[3]  (see \cite{FG06}). This means that putting a problem in \W[1] decreases the
descriptive complexity of the problem.

It is easy to see that \pscsp\ restricted to some finite family of relations (implying the arity of relations is bounded) is in $\W[1]$. 
Notice that listing the tuples of all relations in the input adds just a constant to the size of input.
Thus, uniform and nonuniform \csp\ in this case have the same complexity.
These problems are studied by Marx \cite{Mar05},
where he provides a dichotomy: If the family of relations has a property that he calls \textit{weak-separability},
then the problem is fixed-parameter tractable (like $p$-Vertex-Cover), otherwise it is $\W[1]$-complete (like \pclique).
This result is extended by Bulatov and Marx \cite{BM14} to any finite domain.

There is a variant of \csp\ on the Boolean domain where the variables are weighted. 
That is, each instance comes with a weight function over the variables 
and a target value. It should be decided if there is an assignment that 
satisfies the constraints and the weights of its $1$-variables add up to the target value.
Special cases of this variant are studied in the literature. For example, \cite{DLMR12,TTV17,CS18} study the 
parameterized problem of finding a clique in a weighted graph,  where the parameter is the target value.

In studying this kind of problems, one is in some way dealing with \subsetsum. 
For every computable function $f:\N\rightarrow \N$ we let: 
\begin{center}
\begin{tabular}{ |  l |}
	\hline
	\begin{tabular}{l}
	        \pssf[f] \rule{0pt}{15pt} \\
	\end{tabular}\\
	\begin{tabular}{r l}
		\textit{Instance:} &\makecell[tl]{
					$k\ge 0$, $n$ integers $x_1, \ldots,x_n \in \br{0, n^{f\p{k}}}$,  $t \in \br{0, n^{f\p{k}}}$.
				   }  \\
		\textit{Parameter:} &$k$.\\
		\textit{Problem:} \rule[-15pt]{0pt}{15pt}&\makecell[tl]{Decide whether there exists a subset $B \subset \br{n}$ of size\\
 									$\size{B} = k$ such that $\sum_{i \in B} x_i = t$.}
		\rule[-20pt]{0pt}{20pt}\\
	\end{tabular}\\
\hline
\end{tabular}
\end{center}

\begin{theorem}\label{theorem-subset-sum}
	$\pssf[f]\in \W[1]$, for all computable functions $f$.
\end{theorem}
Our proof gives a \textit{tail-nondeterministic $\kappa$-restricted NRAM program}
deciding the problem.  
Abboud, Lewi and Williams \cite{ALW14} have also proved this theorem.              
Given an instance of the problem, they generate $g\p{k}\cdot n^{o\p{1}}$ instances
of \pclique\ on $n$ node graphs,
such that one of these graphs contains a $k$-clique if and only if the \pss\ instance has a solution.
The proof follows because \pclique\ is $\W[1]$-complete.
Our proof is considerably shorter and arguably more intuitive. 
They also prove that the weighted variant of \pclique\  is in \W[1]. 
Notice that here the parameter is the size of the clique, thus it is a substantially weaker parameterization than that of  \cite{DLMR12,TTV17,CS18}. 
Generalizing this problem in the language of \csp ,  for every computable function $f:\N\rightarrow\N$, we introduce
the \textit{Parameterized Weighted \csp}:
\begin{center}
\begin{tabular}{ |  l |}
\hline
\begin{tabular}{l}
	      \pwcspf[f] \rule{0pt}{15pt}\\ 
\end{tabular}\\
\begin{tabular}{r l}
\textit{Instance:} &\makecell[tl]{
	A set of variables $V$, the domain $\cb{\bm 0, 1}$, 
a set of constraints,\\ list of tuples of each relation, $k \ge 0$,\\
$w:V\rightarrow \br{0, n^{f\p{k}}}$,  $t  \in \br{0, n^{f\p{k}}}$.
}  \\
\textit{Parameter:} &$k$.\\
\textit{Problem:} \rule[-15pt]{0pt}{15pt}&\makecell[tl]{Decide whether there is a satisfying assignment $B$ of size \\     $\size{B}=k$, 
	such that $\sum_{\p{v, 1} \in B} w\p{v} = t$.
\rule[-10pt]{0pt}{10pt}}\\
\end{tabular}\\ 
\hline
\end{tabular}
\end{center}
where $n$ is the size of the input.
\begin{theorem}
	$\pwcspf[f] \in \W[1]$, for all computable functions $f:\N \rightarrow \N$.
\end{theorem}
The proof employs  our proofs of  Theorems \ref{theorem-main} and \ref{theorem-subset-sum}.

For the basic concepts, definitions and notation of  parameterized complexity theory, we refer the reader to \cite{FG06}.
\paragraph{Notation}
{For integers $n$, $m$ with $n \le m$, we let $\br{n, m} \defeq \cb{n, n+1, \ldots, m}$ and $\br{n}\defeq\br{1, n}$.}
\section{A Machine Characterization of \texorpdfstring{\W[1]}{W[1]}}
We use a nondeterministic random access machine model. It is based on a standard deterministic random access machine (RAM) model.
Registers store nonnegative integers. Register $0$ is the \textit{accumulator}.
The arithmetic operations are addition, subtraction (cut off at $0$), and division by two (rounded off), and we use a uniform cost measure.
For more details see \cite{FG06}.
We define a \textit{nondeterministic} RAM, or NRAM, to be a RAM with an additional instruction ``GUESS'' whose semantics is:
\vspace{3pt}

\leftskip=\parindent
\noindent
\textit{Guess a natural number less than or equal to the number stored in the accumulator and store it in the accumulator.}

\vspace{3pt}
\leftskip=0pt

 \noindent
 Acceptance of an input by an NRAM program is defined as usually for non-deterministic machines. Steps of a computation of an NRAM that execute a GUESS  instruction are called \textit{nondeterministic steps}.
\begin{definition}
        Let $\kappa : \Sigma^* \rightarrow \N$ be a parameterization. An NRAM program $P$ is \emph{$\kappa$-restricted} if there are computable functions $f$ and $g$ and a polynomial $p
\p{X}$ such that on every run with input $x \in \Sigma^*$ the program $P$
        \begin{itemize}
                \item performs at most $f\p{k}\cdot p\p{n}$ steps, at most $g\p{k}$ of them being nondeterministic;
                \item uses at most the first $f\p{k}\cdot p\p{n}$ registers;
                \item contains numbers $\le f\p{k}\cdot p\p{n}$ in any register at any time.
        \end{itemize}
        Here $ n \defeq \size{x}$, and $k \defeq \kappa\p{x}$.
\end{definition}
\begin{definition}
        A $\kappa$-restricted NRAM program $P$ is \textit{tail-nondeterministic} if there is a computable function $q$ such
        that for every run of $P$ on any input $x $ all nondeterministic steps are among the last $q\p{\kappa\p{x}}$ steps of the computation.
\end{definition}
The machine characterization of $\W[1]$ reads as follows:
\begin{theorem}[\cite{CFG05}]\label{nram}
	Let $\p{Q, \kappa}$ be a parameterized problem. Then $\p{Q, \kappa} \in \W[1]$ if and only if there is a tail-nondeterministic
        $\kappa$-restricted NRAM program deciding $\p{Q, \kappa}$.
\end{theorem}
\section{Partially Ordered Sets}
The \textit{Möbius function} of a poset $\p{P, \le}$ is a function $\mu:P\times P\rightarrow \mathbb{Z}$  defined recursively as follows. 
\begin{align*}
 &\mu\p{x, y} =
  \begin{cases}
 1  &$x=y$,\\
 -\sum_{x \le z  < y} \mu\p{x,z} & x < y,\\
 0 & x > y.
 \end{cases}
\end{align*}
\begin{theorem}[Möbius inversion formula]
\label{theorem-mobius-inversion-formula}
   Let $\p{P, \le}$ be a finite partially ordered set with a minimum element.
   For functions $f, g:P \rightarrow \mathbb{Z}$, suppose that
   \eqw[mobius-inversion-formula-g]{
   g\p{x}=\sum_{y \le x} f\p{y}.
   }
   Then
   \eq[mobius-inversion-formula-f]{
   f\p{x} = \sum_{y \le x} g\p{y}\mu\p{y, x}.
   }
   Furthermore, there is a computable function $h:\mathbb{N}\rightarrow \mathbb{N}$ such that  
   \eq{
    \size{f\p{x}} \le h(t)\,\max_x \size{g\p{x}} \quad\text{ for all $x$},
   }
    where $t$ is the maximum size of any interval in $P$.
\end{theorem}
\begin{proof}
See {\cite{Rot64}} for a proof of the first claim.
For the bound, it is easy to see that there is a computable function $h':\mathbb{N}\rightarrow \mathbb{N}$ such that  
   \eq{
   \max_{x,y} \size{\mu\p{x,y}} \le h'(t).
   }
Thus, $f\p{x} \le t\; h'\p{t} \; \max_x \size{g\p{x}}$ for all $x$. Set $h\p{t}\defeq t \; h'\p{t}$ and the claim follows.
\end{proof}
\noindent
Let $\p{P, \le}$ be a poset, and $x,y \in P$. We say $y$ covers $x$ if $x < y$ 
and there is no element $z\in P$ such that $x<z<y$. 
Let $Q \subseteq P$. We denote by $\tilde Q$ (with respect to $P$) the set
of all $y\in P$, such that $y\not \in Q$ 
and $y$ covers some $x \in Q$. 

\begin{llemma}\label{lemma-cover}
Let $\p{P, \le}$ be a poset and suppose $Q \subseteq P$  and $y \in P$.  Then
$y \not \in Q$ if and only if all maximal elements of 
$\cb{w \in Q \cup \tilde Q | w\le y}$ are in $\tilde Q$.
\end{llemma}
\begin{proof}
Suppose $y \not \in Q$. If  $y \in Q \cup \tilde Q$, then clearly $y\in \tilde{Q}$ and the claim follows.
Otherwise, let $x < y$ be a maximal element of $\cb{w \in Q \cup \tilde Q | w\le y}$. 
There is at least one $z \in P$ such that $z \le y$ and $z$ covers $x$.
Now, if $x\in  Q$, then either $z \in Q$ or
 $z \in \tilde{Q}$, 
 thus $z\in \cb{w \in Q \cup \tilde Q| w\le y}$,
 contradicting the maximality of $x$. 
Therefore, $x \in \tilde{Q}$.

The other direction is implied by the trivial fact that if $y \in Q$, then 
$y$ is the maximum element of
$\cb{w \in Q \cup \tilde Q| w\le y}$.
\end{proof}
\section{\texorpdfstring{\pscsp}{p-SIZE-CSP}  is in \texorpdfstring{\W[1]}{W[1]}} 
In this section we present a tail-nondeterministic
$\kappa$-restricted NRAM program $Q$ that decides \pscsp .
Given an instance $I_0$ of the problem with the set of variables $V$ and parameter value $k$,
our program first constructs a second instance $I$ with the same set of variables and the same parameter value,
such that $I_0$ and $I$ have the same set of satisfying assignments of size $\le k$, and 
$I$ has the following properties. Each variable appears in each constraint at most once, and for each subset $S \subseteq V$, there is at most one constraint with this set of variables, 
thus each constraint is characterized by its set of variables.
This construction, invented by Papadimitriou and Yannakakis \cite{PY99}, is as follows. 

Henceforth, we characterize an assignment $A$ with the set of 
all $\p{v, d}$, such that $v$ is mapped to $d$ and $d \ne \bm 0$. 

Fix an order on $V$. For each subset $S \subseteq V$, if $I_0$ has a constraint such that the set of variables of the constraint is exactly $S$ (possibly with repetitions),
then $I$ has Relation  $R_S$ of arity $\size{S}$ and Constraint ${\cal C}_S$ defined as follows.
An assignment $A$ of $S$  satisfies ${\cal C}_S$ if and only if
$\size{A} \le k$ and $A$ satisfies every constraint $\cal C$ in $I_0$,
such that the set of variables of $\cal C$  is exactly $S$ (possibly with repetitions).
The order of variables in ${\cal C}_S$ is determined by the order on $V$. Relation $R_S$ is defined accordingly.
Notice that there is a natural bijective mapping of the tuples in $R_S$ to the satisfying assignments of ${\cal C}_S$.

Program $Q$, in its nondeterministic part, guesses an assignment and 
checks if it satisfies all the constraints of $I$. But because $Q$ is tail-nondeterministic $\kappa$-restricted,
we need some method, explained below, instead of trivially going over all constraints.

For a constraint ${\cal C}_S$ in $I$, let $C_S$ be the set of all satisfying assignments of ${\cal C}_S$, 
and $P_S$ be the set of all possible assignments of $S$.
Clearly, $C_S \subseteq P_S$. Thus, we define $\tilde C_S$ with respect to Poset $\p{P_S, \subseteq}$.

Define 
\eqw{
\tilde {C}\defeq \bigcup_{{\cal C}_S \in I} \tilde {C_S},
} and for each $T \in \tilde {C}$ define 
\eqw{
I_T \defeq \cb{{\cal C}_S \in I | T \in \tilde {C_S}}.
}
For each $T$ in $\tilde C$ and each $C_S \in I_T$, let $H_{T,S}$ be the set of proper supersets of $T$ in $C_S$:
\eqw{
H_{T, S}:=\cb{U \in C_S| T \subset  U }.
} 
\begin{llemma}\label{lemma-TB}
Let $B$ be an assignment. 
If $B$ does not satisfy $I$, then 
there is a $T_B \in \tilde C$ such that if $B$ does not satisfy some ${\cal C}_S \in I_{T_B}$ (and
there is at least one such constraint), then $H_{T_B, S}=\emptyset$. 
\end{llemma}
\begin{proof}
If $B$ does not satisfy $I$, then set ${T_B}$ to an element of maximal size in $\tilde C$, such that at least one constraint ${\cal C}_S \in I_{T_B}$ is not satisfied by $B$. By applying Lemma \ref{lemma-cover} to Poset  $\p{P_S, \subseteq}$, $C_S$ and $B$ (restricted to $S$), it follows that all the maximal elements (with respect to inclusion) of $\cb{U \in C_S \cup \tilde {C_S}| U \subseteq B}$ are in $\tilde{C_S}$, and $T_B$ is one of them, otherwise some $T\in \tilde{C_S}$, $T_B \subset T$ is a maximal element, contradicting that $T_B$ has maximal size. This means that $H_{T_B, S}=\emptyset$. 
\end{proof}
\noindent
Lemma \ref{lemma-TB} implies that Program $Q$  has to just check for any $T \in \tilde {C}$, $T \subseteq B$, that for every ${\cal C}_S \in I_T$, there is a $U \in C_S$ that $T \subset U \subseteq B$.
But $Q$ does not have enough \textit{time} to go over all constraints in $I_T$, so the idea is to enumerate all $T \subset U \subseteq B$ and use some \textit{machinery} that (implicitly) adds $1$ for each constraint that passes the check and $0$ for each constraint that fails, to get a number that equals $\size{I_T}$ iff the check is passed for $T$. 
This is not simple, because the same U can be a superset of T in many constraints, 
and T can have many supersets in one constraint. In an earlier work \cite{Maj17}, we applied the inclusion-exclusion principle as the machinery, but with a limited success. In the current paper, we will use the Möbius inversion formula. In fact the Inclusion-Exclusion principal can be proved by Möbius inversion formula \cite{Rot64}, thus can be seen as a special case of it.

We define the function $f_{T, S}:H_{T, S} \rightarrow \mathbb{Z}$ recursively, subject to 
\eq[function-g]{
 \sum_{W \subseteq U} f_{T,S}\p{W} = 1,
}
for all $U \in H_{T, S}$. Applying Theorem \ref{theorem-mobius-inversion-formula} to Poset $\p{H_{T,S}, \subseteq}$ and $f_{T,S}$, it follows that $f_{T, S}$ is evaluated by Möbius inversion formula \eqref{mobius-inversion-formula-f}.
 
 \begin{llemma}
Let $B$ be an assignment. If $B$ satisfies $I$ then for each $T \in \tilde C$, 
\eqw{
\sum_{{\cal C}_S \in I_T} \sum_{\substack{W \in H_{T, S}\\W 
\subseteq B}} f_{T,S}(W) = \size{I_T}.
}
If $B$ does not satisfy $I$, then
\eq{
\sum_{{\cal C}_S \in I_{T_B}} \sum_{\substack{W \in H_{T_B, S}\\W\subseteq B}} f_{T_B,S}(W) < \size{I_{T_B}}.
    }
\end{llemma}
\begin{proof}
If $B$ satisfies a constraint ${\cal C}_S$, then  for all $T \in \tilde {C_S}$, $T \subset B$, it is clear that $B$ is the maximum element of $\cb{W \in H_{T,S}|W 
\subseteq B}$, with respect to inclusion. Thus, by \eqref{function-g}, the inner summation equals $1$.

Now, if $B$ satisfies $I$, then clearly $B$ satisfies all ${\cal C}_S \in I_T$ for each $T\in \tilde C$, and the outer summation equals $\size{I_T}$.
And if $B$ does not satisfy $I$, then  by Lemma \ref{lemma-TB}, for any constraint ${\cal C}_S\in I_{T_B}$ that is not satisfied by $B$, the inner summation is empty and equals $0$. Thus the outer summation is $< \size{I_{T_B}}$.
\end{proof}
\noindent
Now, apply the Fubini’s Principle to swap the summations:
\eqw{
\sum_{\substack{W \in \bigcup_{{\cal C}_
S\in I_T} H_{T, S}\\W \subseteq B}}\quad \sum_{{\cal C}_S \in I_T}  f_{T,S}(W).
}
Program $Q$ evaluates the inner summation  in its deterministic part  and the outer summation in its nondeterministic part, for all $T \in \tilde C$.

Now we are ready for our main theorem:
\begin{theorem}\label{theorem-u-csp-in-w}
\pscsp\  $\in \W[1]$.
\end{theorem}
\begin{proof}
We give a tail-nondeterministic $\kappa$-restricted NRAM program $Q$ deciding the problem. 
The result follows by Theorem \ref{nram}. Let $I_0$ be the given instance with the set of variables $V$ and the parameter value $k$.
Program $Q$ first constructs Instance $I$ from Instance $I_0$ as described above. This can easily be done in polynomial time.

Next, $Q$ calculates two tries:
\vspace{5pt}
\begin{itemize}
\item
Trie 1 stores the values $d\br{T}\defeq \size{I_T}$ for all $T \in \tilde C$.
\item
Trie 2 stores the values 
\eq[definition-function-l]{
        l\br{T, W} \defeq \sum_{{\cal C}_S \in I_T}  f_{T,S}(W),
                }
for every $T \in \tilde C$, and every  $W \in \bigcup_{{\cal C}_
S\in I_T} H_{T, S}$ (if $W\not \in H_{T,S}$, then we extend $f_{T, S}$ to $f_{T, S}\p{W}\defeq 0$).
\end{itemize}
\noindent
For the queries with nonexistent keys, $0$ is returned. 

Now  the nondeterministic part of the computation starts: 
Program $Q$ guesses an assignment $B$ of size $k$.
Then $Q$ iterates over subsets $T \subseteq B$, and if $d\br{T} > 0$, checks if
\eq[equation-L-T_equals_D_T]{
        \sum_{\substack{T\subset W \subseteq B }} l\br{T, W}=d\br{T}.
}
By our argument above, it should be clear that $Q$  decides \pscsp .

We claim that each Trie has at most polynomially many entries. This is because the problem is {uniform} (of size, say, $n$), and the input contains the list of tuples of each relation in $I_0$. Thus, we have
	 $\size{C_S} \le n$ and 
	 $\size{\tilde{C_S}} \le \size{S}\size{D} \p{\size{C_S} + 1} \le n^3$. 

We also claim that $f_{T, S}$ is an fpt-function. This is because by construction of $I$, for all constraints ${\cal C} \in I$, for all $W \in C$, we have $\size{W}\le k$. This implies that  the size of an interval in $H_{T, S}$ is at most $2^k$, and the claim follows by \eqref{function-g} and Theorem \ref{theorem-mobius-inversion-formula}.

The above claims imply that the tries can easily be populated in fpt-time.
 
 Lastly, we show that Program Q is  tail-nondeterministic $\kappa$-restricted. This is because the tries are arranged such that for all assignments $T, U$ of size $\le k$, the query with key $T$ or $T, U$ is answered in
$O\p{k}$ time (a general property of the trie data structure). Moreover, the summation in Check (\ref{equation-L-T_equals_D_T}) has  at most $2^{2k}$ summands. 
This completes the proof.
 \end{proof}

\section{\texorpdfstring{\pwcsp}{p-WCSP} and \texorpdfstring{\pss}{p-SUBSET-SUM} are in \texorpdfstring{\W[1]}{W[1]}} 
In this section we prove that \pwcsp\ and \pss\ are in \W[1]. 
\begin{theorem}\label{theorem-subsetsum}
        $\pssf[f] \in \W[1]$, for all computable functions $f$.
\end{theorem}
\begin{proof}
        For  an integer $a \in \br{0,  n^{f\p{k}}}$, define numbers $\bar a^j \in \br{0, n-1}$ for $j \in \br{0, f\p{k}}$ as
        $a = \sum_{j \in \br{0, f\p{k}}} \bar a^j \, n^{j} $ (this is presentation of $a$ in base $n$, thus $\bar a^j$ are unique).

        For a fixed $f$, we present a tail-nondeterministic $\kappa$-restricted program $P$
 deciding \pssf[f]. On input $x_1, \ldots, x_n, t$, Program $P$ first computes two tables. 
Table $1$  stores the values $\bar x_i^j$ for $i \in \br{n}$ and $j \in \br{0, f\p{k}}$, 
and Table $2$ stores the values $\bar t^j$ for $j \in \br{0, f\p{k}}$. 
The tables are arranged in such a way that the numbers can be accessed in constant time. 
The tables can be easily computed in polynomial time.

        Now the nondeterministic part of the computation starts: Program $P$ guesses $k$ elements in
        $\br{n}$ and checks if they are distinct. Let $B$ be the set of guessed elements.
	Then, for $j \in \br{0, f\p{k}}$, Program $P$ divides $c_{j-1} + \sum_{i \in B } \bar x_i^j $ by $n$, sets $c_j$ as the quotient ($c_{-1} \defeq 0$) and checks if
	the remainder equals $\bar t^j$. Finally, $P$ checks if  $c_{f\p{k}}=0$. 

        Notice that $c_j \le k+1$ for $j \in \br{0, f\p{k}}$, and $P$ can perform each division operations with $O\p{k}$ of its arithmetic operations.        
Thus, the number of steps in the nondeterministic part is $O\p{k f\p{k}}$, and Program $P$ is $\kappa$-restricted tail-nondeterministic.
\end{proof}
\begin{theorem}
	$\pwcspf[f] \in \W[1]$, for all computable functions $f:\N \rightarrow \N$. 
\end{theorem}
\begin{proof}
We present a tail-nondeterministic $\kappa$-restricted program $H$ deciding the problem. 
Let $Q$ be the program that decides \pscsp , and $P$ be the program that decides \pssf[f], as described in the proofs of Theorems \ref{theorem-u-csp-in-w} and \ref{theorem-subsetsum}, respectively. 
$H$ first performs the  deterministic part of $Q$ and then that of $P$. 
Then, $H$  performs the nondeterministic part of $P$: 
it guesses an assignment $B$ of size $\size{B}\defeq k$ 
and checks if weights of variables in $B$ add up to $t$. If no, then
this nondeterministic branch rejects. If yes, then $H$ performs the nondeterministic part of $Q$, omitting the guessing step, 
to check if $B$ is a satisfying assignment. If yes, then $H$ accepts. 
The number of steps in the nondeterministic part of $H$ is bounded by the sum of that of $Q$ and $P$. 
Thus, Program $H$ is $\kappa$-restricted and tail-nondeterministic. 
\end{proof}
\section*{Acknowledgment}
The author thanks Johannes Köbler, Frank Fuhlbrück, and Amir Abboud
for helpful discussions.
\bibliography{kuk-utf8}

\providecommand{\ON}{\ensuremath{\mathcal{O}}}
  \providecommand{\class}[1]{\textsf{\upshape{#1}}}
  \providecommand{\problem}[1]{\textsc{#1}}
  \csname@ifundefined\endcsname{iflongjournalnames}{\newif
  \iflongjournalnames\longjournalnamestrue}\relax
  \providecommand{\journalname}[2]{\iflongjournalnames #1\else #2\fi}
  \providecommand{\proceedings}[3]{\iflongjournalnames Proceedings of #1 #2
  (#3)\else Proc. #1 #3\fi}
  \csname@ifundefined\endcsname{iflongseriesnames}{\newif
  \iflongseriesnames\longseriesnamesfalse}\relax
  \providecommand{\seriesname}[2]{\iflongseriesnames #1\else #2\fi}
\begin{thebibliography}{10}

\bibitem{ALW14}
Amir Abboud, Kevin Lewi, and Ryan Williams.
\newblock Losing weight by gaining edges.
\newblock In Andreas~S. Schulz and Dorothea Wagner, editors, {\em Algorithms -
  {ESA} 2014 - 22th Annual European Symposium, Wroclaw, Poland, September 8-10,
  2014. Proceedings}, volume 8737 of {\em Lecture Notes in Computer Science},
  pages 1--12. Springer, 2014.
\newblock \href {http://dx.doi.org/10.1007/978-3-662-44777-2\_1}
  {\path{https://doi.org/10.1007/978-3-662-44777-2\_1}}.

\bibitem{BM14}
Andrei~A. Bulatov and D{\'{a}}niel Marx.
\newblock Constraint satisfaction parameterized by solution size.
\newblock {\em \journalname{SIAM Journal on Computing}{SIAM J. Comput.}},
  43(2):573--616, 2014.
\newblock \href {http://dx.doi.org/10.1137/120882160}
  {\path{https://doi.org/10.1137/120882160}}.

\bibitem{Ces02}
Marco Cesati.
\newblock Perfect code is {W}[1]-complete.
\newblock {\em \journalname{Information Processing Letters}{Inf. Process.
  Lett.}}, 81(3):163--168, 2002.
\newblock \href {http://dx.doi.org/10.1016/S0020-0190(01)00207-1}
  {\path{https://doi.org/10.1016/S0020-0190(01)00207-1}}.

\bibitem{CFG05}
Yijia Chen, J{\"{o}}rg Flum, and Martin Grohe.
\newblock Machine-based methods in parameterized complexity theory.
\newblock {\em \journalname{Theoretical Computer Science}{Theor. Comput.
  Sci.}}, 339(2-3):167--199, 2005.
\newblock \href {http://dx.doi.org/10.1016/j.tcs.2005.02.003}
  {\path{https://doi.org/10.1016/j.tcs.2005.02.003}}.

\bibitem{DLMR12}
Konrad Dabrowski, Vadim Lozin, Haiko M{\"u}ller, and Dieter Rautenbach.
\newblock Parameterized complexity of the weighted independent set problem
  beyond graphs of bounded clique number.
\newblock {\em Journal of Discrete Algorithms}, 14:207--213, 2012.
\newblock \href {http://dx.doi.org/10.1016/j.jda.2011.12.012}
  {\path{https://doi.org/10.1016/j.jda.2011.12.012}}.

\bibitem{DF95b}
Rod~G. Downey and Michael~R. Fellows.
\newblock Fixed-parameter tractability and completeness {II}: {On} completeness
  for \class{W[1]}.
\newblock {\em \journalname{Theoretical Computer Science}{Theor. Comput.
  Sci.}}, 141(1-2):109--131, 1995.
\newblock \href {http://dx.doi.org/10.1016/0304-3975(94)00097-3}
  {\path{https://doi.org/10.1016/0304-3975(94)00097-3}}.

\bibitem{CS18}
Henri~Perret du~Cray and Ignasi Sau.
\newblock Improved {FPT} algorithms for weighted independent set in bull-free
  graphs.
\newblock {\em Discrete Mathematics}, 341(2):451--462, 2018.
\newblock \href {http://dx.doi.org/10.1016/j.disc.2017.09.012}
  {\path{https://doi.org/10.1016/j.disc.2017.09.012}}.

\bibitem{FV98}
Tom{\'{a}}s Feder and Moshe~Y. Vardi.
\newblock The computational structure of monotone monadic {SNP} and constraint
  satisfaction: {A} study through datalog and group theory.
\newblock {\em {SIAM} J. Comput.}, 28(1):57--104, 1998.
\newblock \href {http://dx.doi.org/10.1137/S0097539794266766}
  {\path{https://doi.org/10.1137/S0097539794266766}}.

\bibitem{FG06}
J{\"{o}}rg Flum and Martin Grohe.
\newblock {\em Parameterized Complexity Theory}.
\newblock Texts in Theoretical Computer Science. An {EATCS} Series. Springer,
  2006.
\newblock \href {http://dx.doi.org/10.1007/3-540-29953-X}
  {\path{https://doi.org/10.1007/3-540-29953-X}}.

\bibitem{KV00}
Phokion~G. Kolaitis and Moshe~Y. Vardi.
\newblock Conjunctive-query containment and constraint satisfaction.
\newblock {\em \journalname{Journal of Computer and System Sciences}{J. Comput.
  Syst. Sci.}}, 61(2):302--332, 2000.
\newblock \href {http://dx.doi.org/10.1006/jcss.2000.1713}
  {\path{https://doi.org/10.1006/jcss.2000.1713}}.

\bibitem{Maj17}
Ruhollah Majdoddin.
\newblock Parameterized complexity of {CSP} for infinite constraint languages.
\newblock {\em CoRR}, abs/1706.10153, 2017.
\newblock URL: \url{http://arxiv.org/abs/1706.10153}, \href
  {http://arxiv.org/abs/1706.10153} {\path{arXiv:1706.10153}}.

\bibitem{Mar05}
D{\'{a}}niel Marx.
\newblock Parameterized complexity of constraint satisfaction problems.
\newblock {\em \journalname{Computational Complexity}{Comput. Complexity}},
  14(2):153--183, 2005.
\newblock \href {http://dx.doi.org/10.1007/s00037-005-0195-9}
  {\path{https://doi.org/10.1007/s00037-005-0195-9}}.

\bibitem{NP85}
Jaroslav Ne{\v{s}}et{\v{r}}il and Svatopluk Poljak.
\newblock On the complexity of the subgraph problem.
\newblock {\em Commentationes Mathematicae Universitatis Carolinae},
  26(2):415--419, 1985.

\bibitem{PY99}
Christos~H. Papadimitriou and Mihalis Yannakakis.
\newblock On the complexity of database queries.
\newblock {\em \journalname{Journal of Computer and System Sciences}{J. Comput.
  Syst. Sci.}}, 58(3):407--427, 1999.
\newblock \href {http://dx.doi.org/10.1006/jcss.1999.1626}
  {\path{https://doi.org/10.1006/jcss.1999.1626}}.

\bibitem{PW10}
Mihai Patrascu and Ryan Williams.
\newblock On the possibility of faster {SAT} algorithms.
\newblock In Moses Charikar, editor, {\em Proceedings of the Twenty-First
  Annual {ACM-SIAM} Symposium on Discrete Algorithms, {SODA} 2010, Austin,
  Texas, USA, January 17-19, 2010}, pages 1065--1075. {SIAM}, 2010.
\newblock \href {http://dx.doi.org/10.1137/1.9781611973075.86}
  {\path{https://doi.org/10.1137/1.9781611973075.86}}.

\bibitem{RBW06}
Francesca Rossi, Peter van Beek, and Toby Walsh, editors.
\newblock {\em Handbook of Constraint Programming}, volume~2 of {\em
  Foundations of Artificial Intelligence}.
\newblock Elsevier, 2006.
\newblock URL:
  \url{http://www.sciencedirect.com/science/bookseries/15746526/2}.

\bibitem{Rot64}
Gian-Carlo Rota.
\newblock On the foundations of combinatorial theory i. theory of m{\"o}bius
  functions.
\newblock {\em Probability theory and related fields}, 2(4):340--368, 1964.
\newblock \href {http://dx.doi.org/10.1007/BF00531932}
  {\path{https://doi.org/10.1007/BF00531932}}.

\bibitem{TTV17}
St{\'{e}}phan Thomass{\'{e}}, Nicolas Trotignon, and Kristina Vuskovic.
\newblock A polynomial turing-kernel for weighted independent set in bull-free
  graphs.
\newblock {\em Algorithmica}, 77(3):619--641, 2017.
\newblock \href {http://dx.doi.org/10.1007/s00453-015-0083-x}
  {\path{https://doi.org/10.1007/s00453-015-0083-x}}.

\end{thebibliography}

\end{document}